%% file: 0-main.tex
\documentclass[11pt]{article}
 
\usepackage[margin=1in]{geometry} 
\usepackage{comment}
\usepackage{amsmath,amsthm,amssymb, color, tikz}
\usetikzlibrary{positioning}
\usetikzlibrary{shapes}
\usetikzlibrary{decorations.pathreplacing}
\usepackage{enumerate, enumitem} \usepackage[noblocks]{authblk}

\usepackage{hyperref} 
\hypersetup{colorlinks,linkcolor={blue},citecolor={blue},urlcolor={red}}
\newcommand{\OR}{\mathsf{OR}}
\newcommand{\AND}{\mathsf{AND}}
\newcommand{\XOR}{\mathsf{XOR}}
\newcommand{\NOR}{\mathsf{NOR}}

\newcommand{\IND}{\mathsf{IND}}
\newcommand{\ADD}{\mathsf{ADDR}}
\newcommand{\HADD}{\mathsf{HADD}}

\newcommand{\PARITY}{\mathsf{PARITY}}
\renewcommand{\DJ}{\mathsf{DJ}}

\newcommand{\R}{\mathbb{R}}
\newcommand{\E}{\mathbb{E}}
\newcommand{\bra}[1]{\{#1\}}
\newcommand{\mathify}[1]{\ifmmode{#1}\else\mbox{$#1$}\fi}
\newcommand{\abs}[1]{\mathify{\big| #1 \big|}}
\newcommand{\pmone}{\bra{-1, 1}}

\renewcommand{\deg}{\mathrm{deg}}
\newcommand{\adeg}{\widetilde{\mathrm{deg}}}
\newcommand{\wh}{\widehat}
\newcommand{\eps}{\epsilon}

\newcommand{\asnorm}[2]{\|\widehat{#1}\|_{1,#2}}
\newcommand{\snorm}[1]{\|\widehat{#1}\|_{1}}

\newcommand{\what}{\textnormal{selector}}

\newtheorem{theorem}{Theorem}[section]
\newtheorem{corollary}[theorem]{Corollary}
\newtheorem{remark}[theorem]{Remark}
\newtheorem{lemma}[theorem]{Lemma}
\newtheorem{claim}[theorem]{Claim}
\newtheorem{definition}[theorem]{Definition}

\newtheorem{conjecture}[theorem]{Conjecture}
\newtheorem{observation}[theorem]{Observation}
\newtheorem{fact}[theorem]{Fact}

\usepackage{bbm}
\usepackage{mathtools}

\DeclareMathOperator*{\Exp}{\mathbb{E}}

\title{Quantum Query-to-Communication Simulation Needs a Logarithmic Overhead}
\date{}

\begin{document}
\author[1]{Sourav Chakraborty}
\author[2]{Arkadev Chattopadhyay}
\author[3]{Nikhil S.~Mande}
\author[1]{Manaswi Paraashar}

\affil[1]{Indian Statistical Institute, Kolkata}
\affil[2]{TIFR, Mumbai}
\affil[3]{Georgetown University}

\maketitle

\begin{abstract}
Buhrman, Cleve and Wigderson (STOC'98) observed that for every Boolean function $f : \pmone^n \to \pmone$ and $\bullet : \pmone^2 \to \pmone$ the two-party bounded-error quantum communication complexity of $(f\circ \bullet)$ is $O(Q(f) \log n)$, where $Q(f)$ is the bounded-error quantum query complexity of $f$. Note that the bounded-error randomized communication complexity of $(f \circ \bullet)$ is bounded by $O(R(f))$, where $R(f)$ denotes the bounded-error randomized query complexity of $f$. Thus, the BCW simulation has an extra $O(\log n)$ factor appearing that is absent in classical simulation. A natural question is if this factor can be avoided. H{\o}yer and de Wolf (STACS'02) showed that for the Set-Disjointness function, this can be reduced to $c^{\log^* n}$ for some constant $c$, and subsequently Aaronson and Ambainis (FOCS'03) showed that this factor can be made a constant. That is, the quantum communication complexity of the Set-Disjointness function (which is $\NOR_n \circ \wedge$) is $O(Q(\NOR_n))$.

Perhaps somewhat surprisingly, we show that when $ \bullet = \oplus$, then the extra $\log n$ factor in the BCW simulation is unavoidable. In other words, we exhibit a total function $F : \pmone^n \to \pmone$ such that $Q^{cc}(F \circ \oplus) = \Theta(Q(F) \log n)$.

To the best of our knowledge, it was not even known prior to this work whether there existed a total function $F$ and 2-bit function $\bullet$, such that $Q^{cc}(F \circ \bullet) = \omega(Q(F))$.
\end{abstract}

\input{1-Introduction}

\input{2-Preliminaries}

\input{3-attempt3b_new}

\input{4-conc}

\bibliography{5-bibo}

\end{document}

%% file: 1-Introduction.tex
\section{Introduction}\label{sec: intro}

Classical communication complexity, introduced by Yao~\cite{Yao79}, is aptly called the `swiss-army-knife' for understanding, especially the limitations of, classical computing. Quantum communication complexity holds the same promise with regards to quantum computing. Yet, there are many problems that remain open. One broad theme is to understand the fundamental differences between classical randomized and quantum protocols, especially for computing total functions. 

Recall a standard way to derive a communication problem from a function $f: \pmone^n \to \pmone$. Each input bit of $f$ is \emph{encoded} between the two players Alice and Bob, using an instance of a binary primitive, denoted by $\bullet : \{-1,1\} \times \{-1,1\} \to \{-1,1\}$, giving rise to the communication problem of evaluating $f \circ \bullet$.
Each input bit to $f$ is obtained by evaluating $\bullet$ on the relevant bit of Alice and that of Bob, i.e. $\big(f \circ \bullet\big)\big(x,y\big) = f\big(\bullet(x_1,y_1),\ldots,\bullet(x_n,y_n)\big)$ and $x,y$ are each $n$-bit strings given to Alice and Bob respectively. 
Many well known functions in communication complexity are derived in this way: Set-Disjointness is $\mathsf{NOR} \circ \wedge$, Inner-Product being $\oplus \circ \wedge$, Equality being $\mathsf{NOR}\circ \oplus$. 
Set-Disjointness is also a standard total function where quantum protocols provably yield a significant cost saving over their classical counterpart.

A natural and well studied question is this regard is what is the relationship between the query complexity of $f$ and the communication problem of $f\circ \bullet$, when the $\bullet$ is $\wedge$ or $\oplus$.\footnote{Here $\wedge$ and $\oplus$ are the AND function and the XOR functions on 2 bits respectively.} This question has been studied for particular interesting functions or special classes of functions. Classically, it is folklore that $$R^{cc}(f \circ \bullet\big) = O(R(f)),$$ where $R(f)$  denotes the bounded-error randomized query complexity of $f$ and $R^{cc}(f\circ \bullet)$ denotes the bounded-error randomized communication complexity  for computing $(f\circ \bullet)(x,y)$ when Alice and Bob has inputs $x$ and $y$ respectively. However in the quantum world such a relation is not at all obvious. In an influential work, Buhrman, Cleve and Wigderson \cite{BCW98} observed that a general and natural recipe exists for constructing a quantum communication protocol for $(f \circ \bullet)$, using a quantum query algorithm for $f$ as a black-box.

\begin{theorem}[\cite{BCW98}]\label{thm: qtm_sim_log_loss}
For any Boolean function $f : \pmone^n \to \pmone$, we have 
\[
Q^{cc}\big(f \circ \bullet\big) = O\big(Q(f)\cdot \log n\big),
\]
where $\bullet$ is either $\wedge$ or $\oplus$.
\end{theorem} 
 Here $Q(f)$ denotes the bounded-error quantum query complexity of $f$, and $Q^{cc}(f\circ \bullet)$ denotes the bounded-error quantum communication complexity for computing $(f\circ \bullet)(x,y)$. Thus in the quantum world one incurs a logarithmic factor in the BCW simulation while no such factor is needed in the randomized setting. The basic question that arises naturally and which we completely answer in this work, is the following: analogous to the classical model, can this multiplicative $\log n$ blow-up in the communication cost be always avoided by designing quantum communication protocols that more cleverly simulate quantum query algorithms? 

A priori, it is not clear what the answer to this question ought to be. For certain special functions and some classes of functions, quantum protocols exist where the $\log n$ factor can be saved. First, H{\o}yer and de Wolf~\cite{HdW02} designed a quantum protocol for Set-Disjointness of cost $O(\sqrt{n}c^{\log^{*}n})$, speeding up the BCW simulation significantly. Later, Aaronson and Ambainis~\cite{AA05} gave a more clever protocol that only incurred a constant factor overhead from Grover's search using more involved ideas.

For partial functions, tightness of the BCW simulation is known in some settings.
For example, consider the Deutsch-Jozsa ($\DJ$) problem, where the input is an $n$-bit string with the promise that its Hamming weight is either 0 or $n/2$, and $\DJ$ outputs $-1$ if the Hamming weight is $n/2$, and 1 otherwise.  $\DJ$ has quantum query complexity $1$ whereas the \emph{exact} quantum communication complexity of $(\DJ\circ \oplus)$ is $\log n$.  Note that it is unclear whether the $\log n$ factor loss here is additive or multiplicative\footnote{Indeed, there are well-known situations where complexity of 1 vs.~$\log n$ can be deceptive. The classical private-coin randomized communication complexity of Equality is $\Theta(\log n)$, whereas the public-coin cost is well known to be $O(1)$. Newman's Theorem shows that this difference in costs, in general, is \emph{not multiplicative} but merely \emph{additive}.}.
Montanaro, Nishimura and Raymond~\cite{MNR11} exhibited a partial function for which the BCW simulation is tight (up to constants) in the \emph{exact} and \emph{non-deterministic} quantum settings.  They also observed the existence of a total function for which the BCW simulation is tight (up to constants) in the unbounded-error setting.
If we allow multi-output partial functions, then tightness of the BCW simulation is known: Consider the following function: $f$ takes an $n$-bit string $x$ as input which is promised to be a Hadamard codeword (see Definition~\ref{defn: hcwd}), and outputs a $\log n$ bit string $z$ for which $x$ is its Hadamard codeword.  The communication problem $f \circ \wedge$, where the inputs to the two players are promised to be such that their bitwise-$\AND$ yields a Hadamard codeword, has bounded-error quantum communication complexity $\log n$, but $Q(f) = O(1)$.  Again, it is not clear here whether the $\log n$ factor loss is additive or multiplicative.

As far as we know, there was no (partial or total) Boolean-valued function $f$ known prior to our work for which the \emph{bounded-error} quantum communication complexity of $f \circ \bullet$ (i.e.~$Q^{cc}(f \circ \bullet)$) is even $\omega(Q(f))$, where $\bullet$ is either $\wedge$ or $\oplus$.

In this paper, we exhibit the first \emph{total function} witnessing the tightness of the BCW simulation in arguably the most well-known quantum model, which is the bounded-error model.

\begin{theorem}\label{thm: mainold}
There exists a total function $F : \pmone^n \to \pmone$ for which,
\begin{equation}\label{eq:mainold}
Q^{cc}(F \circ \oplus) = \Theta(Q(F) \log n).
\end{equation}
\end{theorem}

The statement of Theorem~\ref{thm: mainold} does not necessarily guarantee that a function exists that both satisfies Equation~\ref{eq:mainold} and has bounded-error quantum query complexity (as a function of $n$) arbitrarily close to $n$. We answer this question by proving a more general result, from which Theorem~\ref{thm: mainold} immediately follows.

\begin{theorem}[Main Theorem]\label{thm: extmain}
For any constant $0 < \delta < 1$, there exists a total function $F : \pmone^n \to \pmone$ for which $Q(F) = \Theta(n^{\delta})$ and 
\[
Q^{cc}(F \circ \oplus) = \Theta(Q(F) \log n).
\]
\end{theorem}

\subsection{Overview of our approach and techniques}

To demonstrate the tightness of the BCW simulation for a total function in the quantum bounded-error setting we have to find a function $F$ such that 
$Q^{cc}(F\circ \bullet)=\Theta(Q(F)\log n)$ for some choice of $\bullet$ (that is, either $\bullet$ is $\wedge$ or $\oplus$). This requires us to prove an 
upper bound of $Q(F)$ and a lower bound on $Q^{cc}(F\circ \bullet)$.
We consider the case when $\bullet$ is the $\oplus$ function. 

For the inner function the $\oplus$ function is preferred over the $\wedge$ function for one crucial reason: we have an analytical technique for proving lower bounds on $Q^{cc}(F\circ \oplus)$, due to Lee and Shraibman~\cite{LS09}.
They reduced the problem of lower bounding the bounded-error quantum communication complexity of $(F\circ \oplus)$ to proving lower bounds on an  analytic property of $F$, called its \emph{approximate spectral norm}.
The $\epsilon$-approximate spectral norm of $F$, denoted by $\|\hat{F}\|_{1,\epsilon}$, is defined to be the minimum $\ell_1$-norm of the coefficients of a polynomial that approximates $F$ uniformly to error $\epsilon$ (see Definition~\ref{defn: awt}).
Lee and Shraibman~\cite{LS09} showed that $Q^{cc}(F \circ \oplus) = \Omega(\log \asnorm{F}{1/3})$.
Thus, the lower bound of Theorem~\ref{thm: extmain} follows immediately from the result below.

\begin{theorem}\label{thm: wtdeg}
For any constant $0 < \delta < 1$, there exists a total function $F : \pmone^n \to \pmone$ for which $Q(F) = \Theta(n^{\delta})$ and 
\[
    \log\big( \asnorm{F}{1/3} \big)= \Theta(Q(F) \log n).
\]
\end{theorem} 

There are not many techniques known to bound the approximate spectral norm of a function. This sentiment was expressed both in \cite{LS09} and in the work of Ada, Fawzi and Hatami \cite{AFH12}. On the other hand, classical approximation theory offers tools to prove bounds on a simpler and better known concept called \emph{approximate degree} which has been invaluable particularly for quantum query complexity. The $\epsilon$-approximate degree of $f$, denoted by $\adeg_{\epsilon}(f)$, is the minimum degree required by a real polynomial to uniformly approximate $f$ to error $\epsilon$ (see Definition~\ref{defn: adeg}). Recently, two of the authors \cite{CM17} devised a way of lifting approximate degree bounds to approximate spectral norm bounds. We first show here that technique works a bit more generally, to yield the following: let $\ADD_{m,t} : \pmone^m \to [t]$ be a (possibly partial) addressing function (see Definition~\ref{defn: add}). For any function $f:\pmone^n \to \pmone$, define the (partial) function $f^{\ADD_{m, t}} : \pmone^{n \times t} \times \pmone^{n \times m} \to \pmone$ as follows (formally defined in Definition~\ref{defn: addcomp}):
$$
f^{\ADD_{m, t}}\left(x,y\right) = f\left(x_{1,\ADD_{m, t}\left(y_1\right)}, x_{2,\ADD_{m, t}\left(y_2\right)},\ldots,x_{n,\ADD_{m, t}\left(y_n\right)}\right).
$$

Our main result on lower bounding the spectral norm is stated below.

\begin{lemma}[extending \cite{CM17}]\label{lem: lifting} \label{lemma:degree-norm-lifting}
Let $t > 1$ be any integer, $\ADD_{m,t}$ be any (partial) addressing function and $f : \pmone^n \to \pmone$ be any function. Then,
\[
   \log\left(\asnorm{f^{\ADD_{m,t}}}{1/3}\right) = \Omega\left(\adeg(f) \log t\right).
\]
\end{lemma}
The functions $F$ constructed for the proof of Theorem~\ref{thm: wtdeg} are completions of instances of $\PARITY^{\ADD_{\ell, \ell}}$, and hence Lemma~\ref{lem: lifting} yields lower bounds on the approximate spectral norm of $F$ in terms of the approximate degree of $\PARITY$ (which is known to be maximal).

For the upper bound on $Q(F)$ we use two famous query algorithms. These are Grover's search~\cite{Grover96} and the Bernstein-Vazirani algorithm~\cite{BV97}. The use of these algorithms for upper bounding $Q(F)$ is in the same taste as in the work of Ambainis and de Wolf~\cite{AdW14} although their motivation was quite different than ours. Interestingly, Ambainis and de Wolf used their function to pin down the minimal \emph{approximate degree} of a total Boolean function, all of whose input variables are influential. 

\subsection{Intuition behind the function construction}\label{sec: sketchfn}

\begin{itemize}
    \item From Theorem~\ref{thm: qtm_sim_log_loss} it is known that for all Boolean functions $f,~Q^{cc}(f \circ \oplus) \leq O(Q(f) \log n)$.
    \item In order to prove a matching lower bound, we construct a Boolean function $F$ on $n$ variables such that $\asnorm{F}{1/3} = 2^{\Omega(Q(F) \log n)}$ (Theorem~\ref{thm: wtdeg}). From Theorem~\ref{thm: LS}, this shows that $Q^{cc}(F \circ \oplus) = \Omega(\log \asnorm{F}{1/3}) = \Omega(Q(F) \log n)$.  We want to additionally ensure that $Q(F) = \Theta(n^{\delta})$ for a given constant $0 < \delta < 1$.  A formal definition of $F$ is given in Figure~\ref{fig: cexample}, we attempt to provide an overview on how we arrived at this function below.
    \item Assume $\delta$ is a constant that is least $1/2$, else the argument follows along similar lines by ignoring suitably many input variables when defining the function.
    A natural first attempt is to try to construct a composed function of the form $F = f^{\ADD}$, for some addressing function $\ADD$ (see Definition~\ref{defn: add}) with $\Omega(n^{1 - \delta})$ many target bits, for which $Q(f^{\ADD}) = \Theta(\adeg(f))$.  For the lower bound we use Lemma~\ref{lem: lifting} to show that $\log \asnorm{f^{\ADD}}{1/3} = \Omega(\adeg(f) \log (n^{1 - \delta})) = \Omega(\adeg(f) \log n)$.
    \item Given the upper bound target, we are led to a natural choice of addressing function.
    Let $\HADD_{n^{1 - \delta}}$ be the $(n^{1-\delta}, n^{1-\delta})$-addressing function defined as follows. 
    Fix an arbitrary order on the $n^{1-\delta}$-bit Hadamard codewords (see Definition~\ref{defn: hcwd}) , say $w_1, \dots, w_{n^{1-\delta}}$.
    Define $g$ to be the $\what$ function of $\HADD_{n^{1 - \delta}}$ such that $g(w_i) = i$ for all $i \in [n^{1-\delta}]$, and $g(x) = \star$ for $x \neq w_i$ for any $i \in [n^{1-\delta}]$. 
    \item For any function $f$ on $n^{\delta}/2$ bits, the \emph{partial} function $f^{\HADD_{n^{1 - \delta}}}$ on $n$ inputs has quantum query complexity $O(Q(f) + n^{\delta}/2)$, as we sketch in the next step. We select $f$ appropriately such that this is $\Theta(Q(f))$.  Finally, we define the \emph{total} function $F = f^{\HADD_{n^{1 - \delta}}}$ to be the completion of $f^{\HADD_{n^{1 - \delta}}}$ that evaluates to $-1$ on the non-promised inputs of $f^{\HADD_{n^{1 - \delta}}}$.
    \item We choose the outer function to be $f = \PARITY_{n^{\delta}/2}$ to ensure $Q(F) = \Theta(n^{\delta})$. To prove the upper bound on $Q(F)$, we crucially use the Bernstein-Vazirani and Grover's search algorithms. \begin{itemize}
        \item Run $n^{\delta}/2$ instances of the Bernstein-Vazirani algorithm, one on each block.  This guarantees that if the address variables were all Hadamard codewords, then we would receive the correct indices of the target variables with probability 1 and just $n^{\delta}/2$ queries.
        \item In the next step, we run Grover's search on two $n/2$-bit strings to test whether the output of the first step was correct.  If it was correct, we succeed with probability 1, and proceed to query the $n^{\delta}/2$ selected target variables and output the parity of them.  If it was not correct, Grover's search catches a discrepancy with probability at least $2/3$ and we output $-1$, succeeding with probability at least $2/3$ in this case.
        \item The $n^{\delta}/2$ invocations of the Bernstein-Vazirani algorithm use a total of $n^{\delta}/2$ queries, Grover's search uses another $O(\sqrt{n})$ queries, and the final parity (if Grover's search outputs that the strings are equal) uses another $n^{\delta}/2$ queries, for a cumulative total of $O(n^{\delta} + \sqrt{n}) = O(n^{\delta})$ queries (recall that we assume $\delta \geq 1/2$).
    \end{itemize} 
\end{itemize}

\input{3.5-figure}

\subsection{Other implications of our result}
Zhang~\cite{Zhang09} showed that for all Boolean functions $f$, there must exists gadgets $g_i$, each either $\wedge$ or $\vee$, such that $Q^{cc}(f(g_1, \dots, g_n)) = \Omega(\textnormal{poly} (Q(f)))$.  For monotone $f$, they showed that either $Q^{cc}(f \circ \wedge) = \Omega(\textnormal{poly}(Q(f)))$ or $Q^{cc}(f \circ \OR_2) = \Omega(\textnormal{poly}(Q(f)))$.  They also state that it is unclear how tight the BCW simulation is.  We show that there exists a function for which it is tight up to constants (on composition with $\oplus$).

Another implication of our result is related to the Entropy Influence Conjecture, which is an interesting question in the field of analysis of Boolean functions, posed by Friedgut and Kalai~\cite{FK}. This conjecture is wide open for general functions. A much weaker version of this conjecture is called the Min-Entropy Influence Conjecture. For the statement of the conjecture we need to consider the Fourier expansion of $f:\pmone^n \to \pmone$ as
$$
    f(\textbf{x}) = \sum_{S\subseteq [n]}\widehat{f}(S)\chi_S(\textbf{x}),
$$ 
where $\bra{\chi_S : S \subseteq [n]}$ are the \emph{parity} functions ($\chi_S(\textbf{x}) = \Pi_{i\in S}x_i$, when $\textbf{x} = (x_1, \dots, x_n) \in \pmone^n$) and $\bra{\widehat{f}(S) : S \subseteq [n]}$ are the corresponding Fourier coefficients. 

\begin{conjecture}(Min-Entropy Influence Conjecture) For any Boolean function $f:\pmone^n \to \pmone$ there exists a non-zero Fourier coefficient $\widehat{f}(S)$ such that $$\log\left(1/|\widehat{f}(S)|\right) = O(I(f)),$$
where $I(f)$ denotes the influence (or average sensitivity) of $f$ ($I(f) = \sum_{S\subseteq [n]}|S|\widehat{f}(S)^2$).
\end{conjecture}

While this conjecture is also wide open, some attempts have been made to prove various implications of this conjecture. One interesting implication of the Min-Entropy Influence Conjecture that is still open is whether the min-entropy of the Fourier spectrum (that is, $\log\left(1/\max_{S \subseteq [n]}|\widehat{f}(S)|\right)$) is less than $O(Q(f))$. In \cite{ACK+18} using a primal-dual technique it was shown that the min-entropy of the Fourier spectrum is less than a constant times $\log(\|\hat{f}\|_{1,\epsilon})$, where the constant depends on $\epsilon$.  Thus if it was the case that $\log(\|\hat{f}\|_{1,\epsilon} = O(Q(f)))$, we would have upper bounded the min-entropy of Fourier spectrum by $O(Q(f))$. This was stated in \cite{ACK+18} as a possible approach and was left as an open problem. Our result in this paper nullifies this approach.

%% file: 3.5-figure.tex
\begin{figure}\label{fig: cexample}
\begin{center}
\begin{tikzpicture}

\tikzstyle{gate}=[ellipse,draw=black]
\tikzstyle{input}=[]
\node(hdef) at (-6, -1){$F = $};
\node[gate] (root) at (0,0) {$\mathsf{PARITY}$};
\node[gate] (c1) at (-3,-2) {$\mathsf{HADD}_\ell$};
\node[gate] (cn) at (3,-2) {$\HADD_{\ell}$};
\node (dummy1) at (-1.8, -0.3) {};
\node (dummy2) at (1.9, -0.3) {$k/2$};
\draw[->, dashed] (dummy1) edge [bend right = 20] (dummy2);

\node (dummy3) at (-5, -2.7) {$2\ell$};
\node (dummy4) at (-1, -2.8) {};
\draw[<-, dashed] (dummy3) edge [bend right = 20] (dummy4);

\node[input] (x8) at (1,-2) {$\boldsymbol{\cdot}$};
\node[input] (x19) at (-1,-2) {$\boldsymbol{\cdot}$};
\node[input] (x419) at (0,-2) {$\boldsymbol{\cdot}$};

\node[input] (x11) at (-5.5,-4) {$x_{11}$};
\node[input] (x123) at (-4.5,-4) {$\cdots$};
\node[input] (x1f) at (-3.5,-4) {$x_{1\ell}$};
\node[input] (y11) at (-2.5,-4) {$y_{11}$};
\node[input] (y123) at (-1.5,-4) {$\cdots$};
\node[input] (y1f) at (-0.5,-4) {$y_{1\ell}$};

\node[input] (yn1) at (5.5,-4) {$y_{\frac{k}{2}\ell}$};
\node[input] (x13) at (4.5,-4) {$\cdots$};
\node[input] (ynf) at (3.5,-4) {$y_{\frac{k}{2}1}$};
\node[input] (xn1) at (2.5,-4) {$x_{\frac{k}{2}\ell}$};
\node[input] (y13) at (1.5,-4) {$\cdots$};
\node[input] (xnf) at (0.5,-4) {$x_{\frac{k}{2}1}$};

\draw[->] (c1) -- (root);
\draw[->] (cn) -- (root);

\draw[<-] (c1) -- (y11);
\draw[<-] (c1) -- (y1f);
\draw[<-] (c1) -- (x11);
\draw[<-] (c1) -- (x1f);

\draw[<-] (cn) -- (yn1);
\draw[<-] (cn) -- (ynf);
\draw[<-] (cn) -- (xn1);
\draw[<-] (cn) -- (xnf);

\draw[->] (root) -- ++(0,1);

\draw [
    thick,
    decoration={
        brace,
        mirror,
        raise=0.5cm
    },
    decorate
] (-5.8,-4) -- (-3.2, -4)
node [pos=0.5,anchor=north,yshift=-0.55cm] {Address bits}; 
\draw [
    thick,
    decoration={
        brace,
        mirror,
        raise=0.5cm
    },
    decorate
] (-2.8,-4) -- (-0.3, -4)
node [pos=0.5,anchor=north,yshift=-0.55cm] {Target bits}; 
\end{tikzpicture}
\end{center}
\caption{$k = n^{\delta}, \ell = n^{1 - \delta}$. If the address bits of an input to an $\HADD_\ell$ is the $j$'th Hadamard codeword, then $y_j$ is selected.  If on an input, there exists at least one $\HADD_\ell$ for which the address bits do not correspond to a Hadamard codeword, $F$ outputs $-1$.  Else it outputs the Parity of the $k/2$ selected target bits.}
\end{figure}
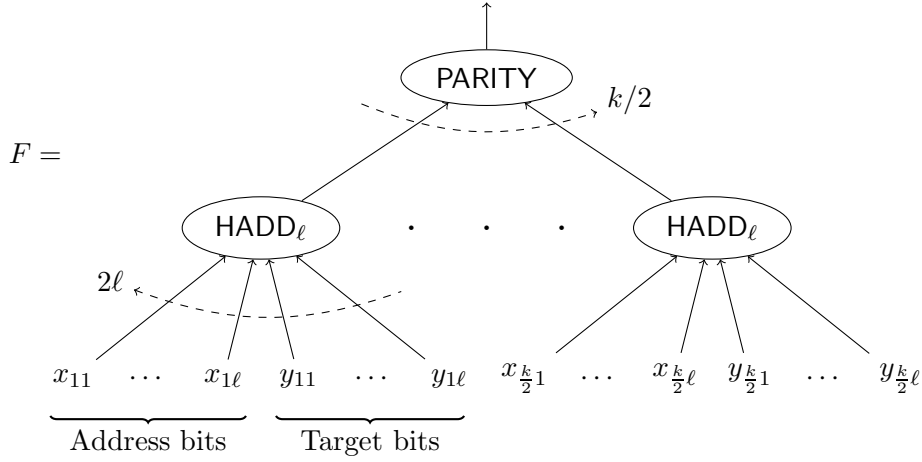

%% file: 2-Preliminaries.tex
\section{Preliminaries}\label{sec: prelims}
For any positive integer $n$, we denote the set $\{1, \dots, n\}$ by $[n]$. For $d \leq n$ we use the notation $\binom{n}{\leq d} = \binom{n}{0} + \cdots + \binom{n}{d}$. Note that $\binom{n}{\leq d} < (n+1)^d$.

In this section we review the necessary preliminaries.  We first review some basics of Fourier analysis on the Boolean cube.
Consider the vector space of functions from $\pmone^n$ to $\R$, equipped with an inner product defined by
\[
\langle f, g \rangle := \E_{x \in \pmone^n}[f(x)g(x)] = \frac{1}{2^n}\sum_{x \in \pmone^n}f(x)g(x).
\]
for every $f,g: \pmone^n \rightarrow \R$. For any set $S \subseteq [n]$, define the associated \emph{parity} function $\chi_S$ by $\chi_S(x) = \prod_{i \in S}x_i$. The set of parity functions $\bra{\chi_S : S \subseteq [n]}$ form an orthonormal basis for this vector space.  Thus, every function $f : \pmone^n \to \R$ has a unique multilinear expression as
\[
f = \sum_{S \subseteq [n]}\wh{f}(S)\chi_S.
\]
The coefficients $\bra{\wh{f}(S) : S \subseteq [n]}$ are called the \emph{Fourier coefficients} of $f$.
\begin{fact}[Parseval's Identity]\label{fact: parseval}
For any function $f : \pmone^n \to \R$,
\[
    \sum_{S \subseteq [n]}\wh{f}(S)^2 = \frac{\sum_{x \in \pmone^n}f(x)^2}{2^n}.
\]
\end{fact}

\begin{definition}[Spectral Norm]\label{defn: wt}
For any function $f : \pmone^n \to \R$, define its \emph{spectral norm}, which we denote $\snorm{f}$, to be the sum of absolute values of the Fourier coefficients of $f$.  That is,
\[
\snorm{f} : = \sum_{S \subseteq [n]}\abs{\wh{f}(S)}.
\]
\end{definition}

\begin{definition}[Hadamard Codeword]\label{defn: hcwd}
If a $\ell$-bit string $(x_1, \dots, x_{\ell}) \in \pmone^\ell$ is of the form $x_S = \prod_{i \in S}z_i$ for all $S \subseteq [\log \ell]$ for some $z \in \pmone^{\log \ell}$, then define such an $x = x_1 \dots x_\ell$ to be the \emph{$\ell$-bit Hadamard codeword} $h(z)$ of the $(\log \ell)$-bit string $z$.
\end{definition}

\subsection{Addressing functions}
\begin{definition}[$(m,k)$-addressing function]\label{defn: add}
We define a (partial) function $f:\pmone^{m+k} \rightarrow \bra{-1, 1, \star}$ to be an \emph{$(m,k)$-addressing function} if there exists $g : \pmone^m \to \bra{[k] \cup \star}$ such that 
\begin{itemize}
    \item $f(x_1, \dots, x_m, y_1, \dots, y_{k}) = y_{g(x_1, \dots, x_m)}$ if $g(x_1, \dots, x_m) \in [t]$, and $f(x_1, \dots, x_m, y_1, \dots, y_{k}) = \star$ otherwise.
    \item For all $j \in [k]$, there exists $(x_1, \dots, x_m) \in \pmone^m$ such that $g(x_1, \dots, x_m) = j$.
\end{itemize}

We call the variables $\{x_1, \dots, x_m\}$ the \emph{address variables} and the variables $\{y_{1}, \dots, y_{k}\}$ the \emph{target variables}.  The function $g$ is called the $\what$ function of $f$.
\end{definition}

\begin{definition}[Indexing Function]
The Indexing function, which we denote by $\IND_k$, is a $(k, 2^k)$-addressing function defined by $\IND(x_1, \dots, x_k, y_1, \dots, y_{2^k}) = y_{\textnormal{bin}(x)}$, where $\textnormal{bin}(x)$ denotes the integer represented by the binary string $x_1, \dots, x_k$.
\end{definition}

\begin{definition}[Composition with addressing functions]\label{defn: addcomp}
For any function $f : \pmone^n \to \pmone$ and an $(m, k)$-addressing function $\ADD$, define the (partial) function $f^{\ADD} : \pmone^{n (m + k)} \to \bra{-1, 1, \star}$ by
\[
f^{\ADD}(x_1, y_1, \dots, x_n, y_n) = 
\begin{cases}
f(\ADD(x_1, y_1), \dots, \ADD(x_n, y_n)) & \text{if~} \ADD(x_i, y_i) \in \pmone \text{~for all~} i \in [n]\\
\star & \text{otherwise}.
\end{cases}
\]
where $x_i \in \pmone^{m}$ and $y_i \in \pmone^k$ for all $i \in [n]$.
\end{definition}

\begin{definition}[Hadamard Addressing Function]\label{defn: hadd}
We define the \emph{Hadamard addressing function}, which we denote $\HADD_\ell : \pmone^{2\ell} \to \bra{-1, 1, \star}$, as follows.
Fix an arbitrary order on the $\ell$-many Hadamard codewords of $(\log \ell)$-bit strings, say $w_1, \dots, w_\ell$.  Define the selector function of $\HADD_\ell$ by
\begin{align*}
g(x) & =\begin{cases}
i & \text{if~} x = w_i \text{~for some~} i \in [\ell]\\
\star & \text{otherwise}.
\end{cases}
\end{align*}
Note that $\HADD_\ell$ is an $(\ell, \ell)$-addressing function.
\end{definition}

\subsection{Polynomial approximation}

\begin{definition}[Approximate Degree]\label{defn: adeg}
The \emph{$\eps$-approximate degree} of a function $f : \pmone^n \to \bra{-1, 1, \star}$, denoted by $\adeg_\eps(f)$ is defined to be the minimum degree of a real polynomial $p : \pmone^n \to \R$ that satisfies $\abs{p(x) - f(x)} \leq \eps$ for all $x \in \pmone^n$ for which $f(x) \in \pmone$.\footnote{When dealing with partial functions, another notion of approximation is sometimes considered, where the approximating polynomial $p$ is required to have bounded values even on the non-promise inputs of $f$.  For the purpose of this paper, we do not require this constraint.} That is,
\[
\adeg_\eps(f) := \min\bra{d : \deg(p) \leq d, \abs{p(x) - f(x)} \leq \epsilon~\text{for all~} x \in \pmone^n~\text{for which}~f(x) \in \pmone}.
\]
\end{definition}
Henceforth, we will use the notation $\adeg(f)$ to denote $\adeg_{1/3}(f)$.

\begin{definition}[Approximate Spectral Norm]\label{defn: awt}
The \emph{approximate spectral norm} of a function  $f : \pmone^n \to \{-1,1,\star\}$, denoted by $\asnorm{f}{\epsilon}(f)$ is defined to be the minimum spectral norm of a real polynomial $p : \pmone^n \to \R$ that satisfies $\abs{p(x) - f(x)} \leq \eps$ for all $x \in \pmone^n$ for which $f(x) \in \pmone$.
\[
\asnorm{f}{\epsilon}(f) := \min\bra{\snorm{p} : \abs{p(x) - f(x)} \leq \epsilon~\text{for all~} x \in \pmone^n~\text{for which}~f(x) \in \pmone}.
\]
\end{definition}
\begin{lemma}[\cite{BNRW07}]\label{lem: asnormequiv} 
Let $f:\pmone^n \to \pmone$ be a total function. Then for all constants $0 < \delta, \eps < 1$ we have $\adeg_\eps(f) = \Theta(\adeg_\delta(f))$.
\end{lemma}

The following is a standard upper bound on the approximate spectral norm of a Boolean function in terms of its approximate degree.
\begin{claim}\label{claim: adegnaivelb}
For all total functions $f : \pmone^n \to \pmone$, we have 
\[
\log \asnorm{f}{1/3}(f) = O(\adeg(f) \log n).
\]
\end{claim}
\begin{proof}
Let $d$ denote the approximate degree of $f$.  Take any 1/3-approximating polynomial of degree $d$, say $p$, to $f$.  Then,
\begin{align*}
    \sum_{S \subseteq [n]} \abs{\wh{p}(S)} & \leq \sqrt{\binom{n}{\leq d}} \cdot \sqrt{\sum_{S: |S| \leq d}\wh{p}(S)^2} \leq 4/3 \cdot (n+1)^{d/2} = 2^{O(d \log n)},
\end{align*}

where the first inequality follows by the Cauchy-Schwarz inequality, the second inequality follows by Parseval's identity (Fact~\ref{fact: parseval}) and the fact that the absolute value of $p$ is at most $4/3$ for any input $x \in \pmone^n$.
\end{proof}

It is easy to exhibit functions $f : \pmone^n \to \pmone$ such that $\log \asnorm{f}{1/3}(f) = \Omega(\adeg(f))$.  Bent functions satisfy this bound, for example.

Building upon ideas in~\cite{KP97}, the approximate spectral norm of $f \circ \IND_1$ was shown to be bounded below by $2^{\Omega(\adeg(f))}$ in~\cite{CM17}.
\begin{theorem}[\cite{CM17}]
Let $f:\pmone^n \rightarrow \pmone$ be any function. Then $\asnorm{f}{1/3}(f \circ \IND_1) \geq 2^{c \cdot \adeg_{2/3}(f)}$ for any constant $c < 1 - 3/\adeg_{2/3}(f)$.
\end{theorem}

\subsection{Communication complexity}
The classical model of communication complexity was introduced by Yao in~\cite{Yao79}. In this model two parties, say Alice and Bob, wish to compute a function whose output depends on both their inputs. Alice is given an input $x \in \mathcal{X}$, Bob is given $y \in \mathcal{Y}$, and they want to jointly compute the value of a given function $F(x,y)$ by communicating with each other.  Alice and Bob individually have unbounded computational power and the number of bits communicated is the resource we wish to minimize. Alice and Bob communicate using a \emph{protocol} that is agreed upon in advance. In the randomized model, Alice and Bob have access to unlimited public random bits and the goal is to compute the correct value of $F(x,y)$ with probability at least $2/3$ for all inputs $(x,y) \in \mathcal{X} \times \mathcal{Y}$. The \emph{bounded-error randomized communication complexity} of a function $F$, denoted $R^{cc}(F)$, is the number of bits that must be communicated in the worst case by any randomized protocol to compute the correct value of the function $F(x,y)$, with probability at least $2/3$, for every $(x,y) \in \mathcal{X} \times \mathcal{Y}$.

The quantum model of communication complexity was introduced by Yao in~\cite{Yao93}. 
We refer the reader to the survey~\cite{Wolf02} for details. The \emph{bounded-error quantum communication complexity} of a function $F$, denoted $Q^{cc}(F)$ is the number of bits that must be communicated by any quantum communication protocol in the worst case to compute the correct value of the function $F(x,y)$, with probability at least $2/3$, for every $(x,y)$ in domain of $F$. Buhrman, Cleve and Wigderson~\cite{BCW98} observed a quantum simulation theorem, which gives an upper bound on the bounded-error quantum communication complexity of a composed function of the form $f \circ \wedge$ or $f \circ \oplus$ in terms of the bounded-error quantum query complexity of $f$ (see Theorem~\ref{thm: qtm_sim_log_loss}).

Lee and Shraibman~\cite{LS09} showed that the bounded-error quantum communication complexity of $f \circ \oplus$ is bounded below by the logarithm of the approximate spectral norm of $f$. Also see~\cite{CM17} for an alternate proof.
\begin{theorem}[\cite{LS09}]\label{thm: LS}
For any Boolean function $f : \pmone^n \to \pmone$,
\[
Q^{cc}(f \circ \oplus) = \Omega(\log \asnorm{f}{1/3}).
\]
\end{theorem}

%% file: 3-attempt3b_new.tex
\section{Proof of Theorem~\ref{thm: extmain}}\label{sec: main}
In this section, we prove Theorem~\ref{thm: extmain}.  We first formally define the function we use.

\subsection{Definition of the function}\label{sec: Fdefn}
If $\delta < 1/2$, then ignore the last $n - 2n^{2\delta}$ bits of the input, and define the following function on the first $2n^{\delta}$ bits of the input. The same argument as in Sections~\ref{sec: upper} and~\ref{sec: lower} give the required bounds for Theorem~\ref{thm: extmain} and Theorem~\ref{thm: wtdeg}.
Hence, we may assume without loss of generality that $\delta \geq 1/2$.

Define the partial function $f : \pmone^n \to \bra{-1, 1, \star}$ by $f = \PARITY_{n^{\delta}/2}^{\HADD_{n^{1 - \delta}}}$.
Define $F$ to be the completion of $f$ that evaluates to $-1$ on the non-promise domain of $f$ (see Figure~\ref{fig: cexample}).

\subsection{Upper bound}\label{sec: upper}
In this section, we prove the following.
\begin{claim}\label{claim: ndelta}
For $F: \pmone^{n} \to \pmone$ defined as in Section~\ref{sec: Fdefn}, we have
\[
Q(F) = \Theta(n^{\delta}).
\]
\end{claim}
The upper bound follows along the lines of a proof in~\cite{AdW14}, and the lower bound just uses the fact that $F$ is at least as hard as $\PARITY_{n^{\delta}/2}$.
\begin{proof}
The following is an $O(n^{\delta})$-query quantum algorithm computing $F$.  For convenience, set $\ell = n^{1-\delta}$ and $k = n^{\delta}$.
Note that since $\delta \geq 1/2$, we have $k = \Omega(\ell)$.
\begin{enumerate}
    \item Run $k/2$ instances of Bernstein-Vazirani algorithm on inputs $(x_{11}, \dots, x_{1\ell}), \dots, (x_{\frac{k}{2}1}, \dots, x_{\frac{k}{2}\ell})$ to obtain $k/2$ strings $z_{i_1}, \dots, z_{i_{k/2}}$.
    \item Run Grover's search~\cite{Grover96, BHMT02} to check equality of the two strings: $h(z_{i_1}), \dots, h(z_{i_\ell})$ and  $x_{11}, \dots, x_{1\ell}, \dots, x_{\frac{k}{2}1}, \dots, x_{\frac{k}{2}\ell}$, i.e.~to check whether the addressing bits of the input are indeed all Hadamard codewords which are output by the first step.
    \item If the step above outputs that the strings are equal, then query the $k/2$ selected variables and output their parity.  Else, output $-1$.
    \begin{itemize}
        \item If the input was indeed of the form as claimed in the first step, then Bernstein-Vazirani outputs the correct $z_{i_1}, \dots, z_{i_\ell}$ with probability 1, and Grover's search verifies that the strings are equal with probability 1.  Hence the algorithm is correct with probability 1 in this case.
        \item If the input was not of the claimed form, then the two strings for which equality is to be checked in the second step are not equal.  Grover's search catches a discrepancy with probability at least $2/3$.  Hence, the algorithm is correct with probability at least $2/3$ in this case.
        \end{itemize}
\end{enumerate}
The correctness of the algorithm is argued above, and the cost is $k/2$ queries for the first step, $O(\sqrt{k\ell})$ queries for the second step, and at most $k/2$ for the third step.
Thus, we have
\[
Q(F) = O(k + \sqrt{k\ell}) = O(k),
\]
since $k = \Omega(\ell)$.  The upper bound in the lemma follows.

For the lower bound, we argue that $F$ is at least as hard as Parity on $k/2$ inputs. To see this formally, set all the address variables such that the selected target variables are the first target variable in each block.  Under this restriction, $F$ equals $\mathsf{PARITY}(y_{11}, \dots, y_{\frac{k}{2}1})$.  Thus any quantum query algorithm computing $F$ must be able to compute $\mathsf{PARITY}_{k/2}$, and thus $Q(F) \geq k/2$.
\end{proof}

\begin{remark}\label{rmk: outer}
The same argument as above works when the function $f$ is defined to be $g^{\HADD_\ell}$ for any $g : \pmone^{n^{\delta}} \to \pmone$ satisfying $\adeg(g) = \Omega(n^{\delta})$, and $F$ is the completion of $f$ that evaluates to $-1$ on all non-promise inputs.  The same proof of Theorem~\ref{thm: extmain} also goes through, but we fix $g = \PARITY_{n^{\delta}/2}$ for convenience.
\end{remark}

\subsection{Lower bound}\label{sec: lower}

In this section, we first prove Lemma~\ref{lem: lifting}.
We require the following observation.
\begin{observation}\label{obs: obs}
For any monomial $\chi_S,~S\subseteq [n]$ and any $j\in S$, we have 
\[
\Exp_{x_j \sim \pmone}[\chi_S(x)] = 0,
\]
where $x_j$ is distributed uniformly over $\pmone$.
\end{observation}

\begin{proof}[Proof of Lemma~\ref{lem: lifting}]

Let $F = f^{\ADD_{m, t}}$. Recall that our goal is to show that $\log \asnorm{F}{1/3} = \Omega(\adeg(f) \log t)$. We may assume $\adeg(f) \geq 1$, because the lemma is trivially true otherwise.  

Towards a contradiction, suppose there exists a polynomial $P$ of spectral norm strictly less than $\left(\frac{1}{10}\adeg_{0.99}(f) \log t\right)$ uniformly approximating $F$ to error $1/3$ on the promise inputs (recall that from Lemma~\ref{lem: asnormequiv}, we have $\adeg(f) = \Theta(\adeg_{0.99}(f))$).

Let $\nu$ be a distribution on the address bits of $\ADD_{m, t}$ such that $\nu$ is supported only on assignments to the address variables that do not select $\star$, and is the uniform distribution over these assignments.
Let $\mu = \nu^n$ be the product distribution over the address bits of the addressing functions in $F$.

\begin{itemize}
    \item For any assignment $z$ of the address variables from the support of $\mu$, define a relevant (target) variable to be one that is selected $z$. Analogously, define a target variable to be irrelevant if it is not selected by $z$. Define a monomial to be relevant if it does not contain irrelevant variables, and irrelevant otherwise.
    \item Note that for any target variable, the probability with which it is selected is exactly $1/t$.
    \item Thus under any assignment $z$ drawn from $\mu$, for any monomial of the restricted function $P|_z$ of degree $t \geq \adeg_{0.99}(f)$, the probability that it is relevant is at most $1/t^{\adeg_{0.99}(f)}$.
    \item Hence, 
    \begin{align*}
    & \Exp\limits_{z\sim\mu}[\text{$\ell_1$-norm of relevant monomials in } P|_{z} \text{ of degree } \geq \adeg_{0.99}(f)]\\
    &= \sum\limits_{|S| \geq \adeg_{0.99}(f)} |w_S| \Pr\limits_{z\sim\mu}[\chi_S \text{ is relevant}] \\
    &\leq \Pr\limits_{z\sim\mu}[\chi_S \text{ is relevant}] \cdot \snorm{P} \\
    &< \frac{1}{t^{\adeg_{0.99}(f)}} \cdot 2^{\frac{1}{10}\adeg_{0.99}(f) \log t} \\
    & = 2^{(-\frac{9}{10})\adeg_{0.99}(f)\log t}< \frac{3}{5},
    \end{align*}
    where the last inequality holds because $t\geq 2$ and $\adeg_{0.99}(f) \geq 1$.
    \item Fix an assignment to the address variables from the support of $\mu$ such that under this assignment, the $\ell_1$-norm of the relevant monomials in $P$ of degree $\geq \adeg_{0.99}(f)$ is less than $3/5$.
    \item Note that under this assignment (in fact under any assignment in the support of $\mu$), the restricted $F$ is just the function $f$ on the $n$ variables selected by the addressing functions. Denote by $P_1$ the polynomial on the target variables obtained from $P$ by fixing address variables as per this assignment.
    \item Drop the relevant monomials of degree $\geq \adeg_{0.99}(f)$ from $P_1$ to get a polynomial $P_2$, which uniformly approximates the restricted $F$ (which is $f$ on $n$ variables) to error $1/3 + 3/5 < 0.99$.
    \item Take expectation over irrelevant variables (from the distribution where each irrelevant variable independently takes values uniformly from $\pmone$).  Under this expectation, the value of $F$ does not change (since irrelevant variables do not affect $F$'s output by definition), and all irrelevant monomials of $P_2$ become 0 (using Observation~\ref{obs: obs} and linearity of expectation).  Hence, under this expectation we have $\Exp[P_2] = P_3$, where $P_3$ is a polynomial of degree strictly less than $\adeg_{0.99}(f)$.  Furthermore, $P_3$ uniformly approximates $f$ to error less than $0.99$ which is a contradiction. 
\end{itemize}
\end{proof}

As a corollary of Lemma~\ref{lem: lifting}, we obtain a lower bound on the approximate spectral norm of $F$, where $F$ is defined as in Section~\ref{sec: Fdefn}.  This yields a proof of Theorem~\ref{thm: wtdeg}.

\begin{proof}[Proof of Theorem~\ref{thm: wtdeg}]
Given $0 < \delta < 1$, construct $F$ as in Section~\ref{sec: Fdefn}.  Claim~\ref{claim: ndelta} implies 
\[
Q(F) = \Theta(n^{\delta}).
\]
Let $f = \PARITY_{n^{\delta}/2}^{\HADD_{n^{1 - \delta}}}$.  Lemma~\ref{lem: lifting} implies that 
\[
\asnorm{f}{1/3} = \Omega(n^{\delta} \log n).
\]
Since $F$ is a completion of $f$, we have 
\[
\asnorm{F}{1/3} = \Omega(n^{\delta} \log n),
\]
which proves the lower bound in Theorem~\ref{thm: wtdeg}.  The upper bound follows from Theorem~\ref{thm: qtm_sim_log_loss}.
\end{proof}

We are now ready to prove our main theorem.
\begin{proof}[Proof of Theorem~\ref{thm: extmain}]
It immediately follows from Theorem~\ref{thm: wtdeg} and Theorem~\ref{thm: LS}.
\end{proof}

%% file: 4-conc.tex
\section{Conclusions}

We conclude with the following points: first, we find our main result somewhat surprising that simulating a query algorithm by a communication protocol in the quantum context has a larger overhead than in the classical context. Second, it is remarkable that this relatively fine overhead of $\log n$ can be detected using analytic techniques that are an adaptation of the generalized discrepancy method. Third, the function that we used in this work is an $\XOR$ function. Study of this class of functions is proving to be very insightful. A recent example is the refutation of the log-approximate-rank conjecture~\cite{CMS19} and even its quantum version~\cite{ABT18,SdW18}. Our work further advocates the study of this rich class. 

An open question that remains is whether there exists a Boolean function $F : \pmone^n \to \pmone$ such that $Q^{cc}(F \circ \wedge) = \Omega(Q(F) \log n)$.  Or does there exist a better quantum communication protocol for $(F \circ \wedge)$ that does not incur the logarithmic factor loss?

Along with the fact that $\adeg(f) \leq 2Q(f)$~\cite{BBC+01}, Theorem~\ref{thm: wtdeg} yields the following corollary.
\begin{corollary}\label{cor: wtdeg}
For any constant $0 < \delta < 1$, there exists a total function $F : \pmone^n \to \pmone$ for which $\adeg(F) = O(n^{\delta})$ and
\[
\log \asnorm{F}{1/3} = \Omega(\adeg(F) \log n).
\]
\end{corollary}
It is easy to verify that the constructions of $F$ that yield Theorem~\ref{thm: wtdeg} for any fixed constant $0 < \delta < 1$, also satisfy $\adeg(F) = \Theta(n^{\delta})$. Thus, Corollary~\ref{cor: wtdeg} also gives a negative answer to Open Problem 2 in~\cite[Section 6]{ACK+18}, where it was asked if any degree-$d$ approximating polynomial to a Boolean function of approximate degree $d$ has weight at most $2^{O(d)}$. Thus to prove min-entropy of the Fourier spectrum of a Boolean function is upper bounded by approximate degree, it cannot follow from their observation that min-entropy is upper bounded by the logarithm of the approximate spectral norm. It remains an interesting and important open problem: (how) can one prove that the min-entropy of the Fourier spectrum of a Boolean function is upper bounded by a constant multiple of its approximate degree? This inequality is an implication of the Fourier Entropy Influence (FEI) Conjecture. 

\section{Acknowledgements}

We thank Ronald de Wolf and Srinivasan Arunachalam for various discussions on this problem. They provided us with a number of important pointers that were essential for our understanding of the problem and its importance.